\documentclass[12pt,notitlepage]{article}%
\usepackage{amsthm}
\usepackage{setspace}
\usepackage{eurosym}
\usepackage{subfigure}
\usepackage{xcolor}
\usepackage{amssymb}
\usepackage{mathrsfs}
\usepackage{graphicx}
\usepackage[colorlinks,linkcolor=links,citecolor=cites,urlcolor=MyDarkBlue]%
{hyperref}
\usepackage{amsfonts}
\usepackage{amsmath}
\usepackage{amstext}
\usepackage{appendix}
\usepackage{appendix}
\usepackage{mathpazo}
\usepackage{tikz}
\usepackage{verbatim}

\usetikzlibrary{trees}
\usetikzlibrary{matrix}
\usetikzlibrary{decorations.pathreplacing}
\usetikzlibrary{topaths,calc}
\usepackage{natbib}%
\providecommand{\U}[1]{\protect\rule{.1in}{.1in}}
\newtheorem{theorem}{Theorem}

\newtheorem{definition}{Definition}
\newtheorem{example}{Example}

\newtheorem{proposition}{Proposition}

\numberwithin{equation}{section}

\definecolor{MyDarkBlue}{rgb}{0,0.08,0.45}
\definecolor{cites}{HTML}{324b13}
\definecolor{links}{HTML}{1a663b}
\definecolor{MyLightMagenta}{cmyk}{0.1,0.8,0,0.1}
\hypersetup{
colorlinks,citecolor=blue,filecolor=black,linkcolor=blue,urlcolor=blue
}
\usepackage[top=0.9in,bottom=0.9in,left=0.9in,right=0.9in]{geometry}

\usepackage{setspace}
\usepackage{footmisc}
\usepackage{lipsum}

\begin{document}
\title{Firm-worker hypergraphs}
\author{Chao Huang\thanks{Institute for Social and Economic Research, Nanjing Audit University. Email: huangchao916@163.com.}}
\date{}
\maketitle

\begin{abstract}
A firm-worker hypergraph consists of edges in which each edge joins a firm and its possible employees. We show that a stable matching exists in both many-to-one matching with transferable utilities and discrete many-to-one matching when the firm-worker hypergraph has no nontrivial odd-length cycles. Firms' preferences satisfying this condition arise in a problem of matching specialized firms with specialists.
\end{abstract}



\section{Introduction}\label{Sec_intro}

The problem of two-sided matching considers how to match two groups of agents, such as men and women, workers and firms, and students and schools. A key solution concept is stable matching, which excludes incentives for agents to rematch in the market. The literature on two-sided matching has developed into two parallel lines. The problem of matching with transferable utilities originated with the assignment problem (e.g., \citealp{KB57} and \citealp{SS71}). This literature assumes continuous monetary transfers between two market sides and agents have transferable utilities. We can also view a many-to-one matching market with continuous transfers as an exchange economy in which workers sell labors to firms. When agents' valuations are gross-substitute, stable matchings exist in two-sided matching (\citealp{KC82}), and equilibria exist in exchange economies with indivisible goods (\citealp{GS99}). The discrete matching problem originated with the stable marriage problem (\citealp{GS62}). This literature assumes no monetary transfers, discrete wages, or discrete contracts between two market sides. When firms have substitutable preferences, stable matchings exist in discrete matching (\citealp{RS90}), matching with contracts (\citealp{HM05}), and matching with externalities (\citealp{PY23}).

In many real-life markets, complementarities are prevalent and have become a critical issue for market practice. In an exchange economy with transferable utilities, \cite{GS99} showed that any condition ensuring an equilibrium and including unit demands cannot allow complementarities. \cite{HK08} showed that this is also the case for discrete matching without contracts. In the context of matching with contracts, their bilateral and unilateral substitutes conditions (\citealp{HK10}) are imposed on individuals and guarantee the existence of stable matchings. However, in basic settings of both matching with transferable utilities and discrete matching without contracts, we cannot expect a general condition imposed on individuals and compatible with complementarities. In this paper, we provide a condition on the structure of firms' possible sets of employees for basic settings of matching with transferable utilities and discrete matching. Let us illustrate with two examples below. Consider the following discrete matching market from \cite{CKK19} with two firms $f_1,f_2$, and two workers $w_1,w_2$.
\begin{equation}\label{exam_in1}
\begin{aligned}
&f_1: \{w_1,w_2\}\succ\emptyset \qquad\qquad\qquad\qquad &w_1: &\quad f_1\succ f_2\\
&f_2: \{w_1\}\succ\{w_2\}\succ\emptyset \qquad\qquad\qquad\qquad &w_2: &\quad f_2\succ f_1
\end{aligned}
\end{equation}
This market has no stable matchings. If firm $f_1$ hires both workers, worker $w_2$ will go to firm $f_2$; then, firm $f_1$ will drop worker $w_1$; and then, firm $f_2$ will drop worker $w_2$ and hire worker $w_1$; but we then return to the beginning: both workers will go to firm $f_1$. This cyclic blocking process coincides with the cycle in the following \textbf{firm-worker hypergraph}. A hypergraph is a generalization of a graph in which an edge can join any number of vertices. The firm-worker hypergraph uses edges to join each firm with its possible employees.
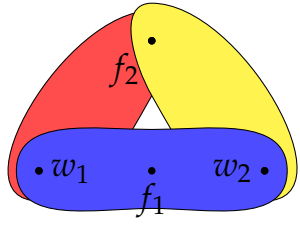
\begin{center}
\begin{tikzpicture}[scale=0.5]
    \node (v1) at (-1,0) {};
    \node (v2) at (2,3.46) {};
    \node (v3) at (5,0) {};
    \node (v4) at (2,2) {};
    \node (v5) at (2,-1) {};
    \node (v6) at (2,0) {};

    \begin{scope}[fill opacity=0.8]
    \filldraw[fill=red!70] ($(v2)+(0,0.8)$)
        to[out=0,in=60] ($(v4)$)
        to[out=240,in=0] ($(v1) + (0,-0.8)$)
        to[out=180,in=180] ($(v2) + (0,0.8)$);
    \filldraw[fill=yellow!80] ($(v2)+(0,1)$)
        to[out=180,in=120] ($(v4)$)
        to[out=300,in=180] ($(v3) + (0,-0.8)$)
        to[out=0,in=0] ($(v2) + (0,1)$);
    \filldraw[fill=blue!70] ($(v1)+(-0.6,0)$)
        to[out=90,in=180] ($(v4)+ (0,-0.9)$)
        to[out=0,in=90] ($(v3) + (0.6,0)$)
        to[out=270,in=0] ($(v5)$)
        to[out=180,in=270] ($(v1) + (-0.6,0)$);
    \end{scope}

    \foreach \v in {1,2,3} {
        \fill (v\v) circle (0.1);
    }

    \fill (v1) circle (0.1) node [right] {$w_1$};
    \fill (v2) circle (0.1) node [below left] {$f_2$};
    \fill (v3) circle (0.1) node [left] {$w_2$};
    \fill (v6) circle (0.1) node [below] {$f_1$};
\end{tikzpicture}
\end{center}
\medskip

We then consider the market in which firms and workers have transferable utilities and bargain over continuous wages.  Consider a parallel market of the above one. Firm $f_1$ derives $6$ from hiring both workers and does not want to hire only one worker. Firm $f_2$ values each worker at $4$ in the absence of the other but does not want to hire both. Both workers care only about their wages. This market is adapted from an exchange economy used as illustrative examples in \cite{AWW13} and \cite{BK19}. There are no stable matchings in this market either. If firm $f_1$ hires both workers at an expenditure of no more than 6, then there must be a worker earning a wage no more than 3. This worker will go to firm $f_2$ since firm $f_2$ is willing to pay a wage between 3 and 4, then the other worker will be fired by firm $f_1$ and come to compete for firm $f_2$'s offer. The competition will lead to a market wage of less than 3, and we then return to the beginning: Firm $f_1$ will hire both workers at an expenditure of no more than 6. This market also has the above firm-worker hypergraph.

We call a cycle in the firm-worker hypergraph a nontrivial odd-length cycle if it has an odd number of edges and each edge contains exactly two vertices of the cycle. This paper shows that a stable matching exists in both discrete matching and matching with transferable utilities if the firm-worker hypergraph has no nontrivial odd-length cycles. A hypergraph without nontrivial odd-length cycles is called balanced. This notion is proposed by \cite{B70}. Whether a hypergraph is balanced or not can be tested in polynomial time; see \cite{CCR99}. The above firm-worker hypergraph fails this condition.\footnote{Notice that vertex $f_1$ is not a vertex of the cycle; see Section \ref{Sec_Mct} for the formal definition.}

Our condition is compatible with complementarities. We provide an application of our existence theorems. In many real-life sectors, a firm often concentrates on one technology path; colleges and firms also train students and workers to be specialists rather than generalists. We study a problem of matching specialized firms with specialists, where specializations are captured in a structure called \textbf{technology roadmap}. A technology roadmap is a tree of technologies in which each technology demands a set of workers to implement. Each firm possesses some technologies and has a preference over its technologies, which induces the firm's preference ordering over the sets of workers demanded by the technologies. Workers are called specialists according to the scopes of technologies they undertake. Firms are called specialized according to the scopes of technologies they possess. We show that when firms and workers are specialized in specific forms, the firm-worker hypergraph has no nontrivial odd-length cycles. Hence, stable matchings exist in related matching markets with transferable utilities and discrete matching markets.

\subsection{Related literature}\label{Sec_lit}

Both our results in discrete matching and matching with transferable utilities are due to a property of balanced matrices proposed by \cite{FHO74}. In matching with transferable utilities, the set of stable matchings coincides with the core. Our result in this market follows from the linear programming duality for the Bondareva-Shapley theorem (\citealp{B63} and \citealp{Sh67}). A matching market with transferable utilities can be viewed as a special exchange economy. \cite{BM97} and \cite{TY19} used linear programs to characterize the existence of equilibria in exchange economies.\footnote{\cite{KC82} demonstrated the equivalence between equilibria and stable matchings in matching with transferable utilities. A previous version of our paper used this equivalence and the linear program of \cite{TY19} to prove our result. The current version uses an argument more direct for our purpose.} The Bondareva-Shapley balancedness condition has also been used by \cite{M98} to characterize the existence of equilibria.

Our result in discrete matching follows from the integer feasibility program of \cite{H23a} (henceforth H23a), which is based on the existence theorem of \cite{CKK19}. The latter theorem implies the existence of a stable schedule matching. Our condition guarantees the transformation of a stable schedule matching into a stable (full-time) matching. This method closely relates to that of \cite{NV18,NV19}, who used Scarf's lemma to round a stable fractional matching into a stable integral one of a slightly different market. Our result is also in a similar form as the no-odd-party condition of \cite{T91} in the stable roommate problem and the acyclicity condition of \cite{BH21} in multilateral matching.

Since the studies of \cite{KC82} and \cite{RS90}, parallel existence results have been found in matching with transferable utilities and discrete matching. See, for example, the within-group substitutability and cross-group complementarity studied by \cite{SY06}, \cite{O08}, and \cite{HKNOW13,HKNOW21}, the complementary contracts studied by \cite{RY20}, and the unimodularity conditions studied by \cite{DKM01}, \cite{BK19}, and \cite{H23a}. Our results are independent of previous conditions since our condition only concerns which sets of workers are acceptable to firms whereas previous conditions are related to agents' values or preference orderings of their acceptable sets.

\subsection{Cyclic structure and cyclic blocking process}\label{Sec_cyclic}

Because our analysis is based on sophisticated tools, our proofs do not provide a direct explanation for the existence theorems. Nonetheless, the following argument may provide some intuitions.
\begin{equation}\label{cycle}
\begin{aligned}
&\qquad\text{When no stable matchings exist in a market,\, the agents will}\\
&\text{constantly block the current matching and never arrive at a static}\\
&\text{state.\, Since the possible states of a market are finite,\, the market}\\
&\text{will return to the same state from time to time. Intuitively, a cyclic}\\
&\text{blocking\, process\, should\, correspond\, to\, a\, cyclic\, structure\, in\, the}\\
&\text{preference profile.}
\end{aligned}
\end{equation}
The above two examples illustrate this argument. Our results further imply that neither trivial odd-length cycles\footnote{See cycle (\ref{odd_cycle}) of Example \ref{exam_A} for a trivial odd-length cycle.} nor even-length cycles cause instability. The following example may help us understand this phenomenon. Consider a marriage market with two men $m_1,m_2$, two women $w_1$,$w_2$, and the following preferences.
\begin{equation}\label{exam_in2}
\begin{aligned}
&m_1: \quad w_1\succ w_2 \qquad\qquad\qquad\qquad &w_1: \quad m_2\succ m_1\\
&m_2: \quad w_2\succ w_1 \qquad\qquad\qquad\qquad &w_2: \quad m_1\succ m_2
\end{aligned}
\end{equation}
The firm-worker hypergraph reduces to the following woman-man graph if we view women as firms and men as workers.
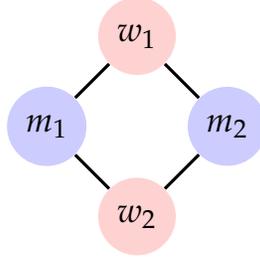
\begin{center}
\tikzstyle{man}=[circle,fill=blue!20]
\tikzstyle{woman}=[circle,fill=pink!70]
\begin{tikzpicture}[scale=0.5]

  \node (m1) at (-2.4,0) [man] {$m_1$};
  \node (m2) at (2.4,0) [man] {$m_2$};
  \node (w1) at (0,2.4) [woman] {$w_1$};
  \node (w2) at (0,-2.4) [woman] {$w_2$};

  \foreach \from/\to in {w1/m1,w1/m2}
    \draw (\from) -- (\to) [very thick];

  \foreach \from/\to in {w2/m1,w2/m2}
    \draw (\from) -- (\to) [very thick];

\end{tikzpicture}
\end{center}
There is a cycle of length 4, which gives rise to the following cyclic blocking process: Woman $w_1$ divorces man $m_1$ and marries man $m_2$; man $m_2$ divorces woman $w_1$ and marries woman $w_2$; woman $w_2$ divorces man $m_2$ and marries man $m_1$; man $m_1$ divorces woman $w_2$ and marries woman $w_1$. However, unlike the above two examples, this cycle breaks into a stable marriage in which woman $w_1$ marries man $m_2$, and woman $w_2$ marries man $m_1$ (or woman $w_1$ marries man $m_1$, and woman $w_2$ marries man $m_2$).

One may consider whether our results can be enhanced into the following statement: Stable matchings exist in all markets with a specific firm-worker hypergraph if and only if the firm-worker hypergraph has no nontrivial odd-length cycles. This statement is similar as the characterizations of \cite{BH21} on venture structures in multilateral matching that guarantee outcomes satisfying different stability concepts. The necessity part means that, given a firm-worker hypergraph containing a nontrivial odd-length cycle, there is always a market with this firm-worker hypergraph for which a stable matching does not exist. The necessity part is plausible because, as we have illustrated, an odd-length cycle in the firm-worker hypergrah may cause a blocking process that does not break down. However, we provide an example in the Appendix (see Section \ref{necessity}) to show that this necessity part does not hold for our condition.

\section{Matching with transferable utilities\label{Sec_Mct}}

In this section, we study many-to-one matching with continuous transfers between firms and workers. We assume each firm and worker has a transferable utility.

There is a finite set $F$ of firms and a finite set $W$ of workers. Define $N\equiv F\cup W$. Let ${\o}$ be the null firm, representing not being matched with any firm; and let $\widetilde{F}\equiv F\cup\{{\o}\}$. Let $\mathcal{A}_f\subseteq 2^W\setminus\{\emptyset\}$ be the collection of firm $f$'s \textbf{acceptable sets of workers}. Each firm $f\in F$ has a valuation $\mathrm{v}_f:\mathcal{A}_f\cup\{\emptyset\}\rightarrow \mathbb{R}$. Let $A_w\subseteq F$ be the collection of firm $f$'s \textbf{acceptable firms}. Each worker $w\in W$ has a valuation $\mathrm{v}_w:A_w\cup\{ø\}\rightarrow \mathbb{R}$. For normalization, we assume $\mathrm{v}_f(\emptyset)=\mathrm{v}_w({\o})=0$ for each $f\in F$ and each $w\in W$. Firms and workers have transferable utilities: Given a price $\mathbf{p}\in \mathbb{R}^W$ where $\mathrm{p}(w)$ is the salary for hiring worker $w$, firm $f$'s utility of hiring $S\in \mathcal{A}_f\cup\{\emptyset\}$ is given by $\mathrm{U}_f(S,\mathbf{p})\equiv \mathrm{v}_f(S)-\sum_{w\in S}\mathrm{p}(w)$ for each $f\in F$, and worker $w$'s utility of working for $f\in A_w\cup\{ø\}$ is given by $\mathrm{U}_w(f,\mathbf{p})\equiv \mathrm{v}_w(f)+\mathrm{p}(w)$ for each $w\in W$. Let $\mathrm{v}_F$ and $\mathrm{v}_W$ be the collections of firms' and workers' valuations, repectively. A matching market with transferable utilities can be summarized as a tuple $\Phi=(F,W,\mathrm{v}_F,\mathrm{v}_W)$. A matching is an assignment of workers to firms together with a wage schedule.

\begin{definition}
\normalfont
A \textbf{matching} $(\mu,\mathbf{p})$ consists of a function $\mu:\widetilde{F}\cup W\rightarrow\widetilde{F}\cup 2^W$ and a price $\mathbf{p}\in \mathbb{R}^W$ such that for all $f\in \widetilde{F}$ and all $w\in W$,
\begin{description}
\item[(\romannumeral1)] $\mu(w)\in \widetilde{F}$, and $\mu(w)={\o}$ implies $\mathrm{p}(w)=0$;

\item[(\romannumeral2)] $\mu(f)\in 2^W$;

\item[(\romannumeral3)] $\mu(w)=f$ if and only if $w\in\mu(f)$.
\end{description}
\end{definition}

The concept of stability includes two requirements. First, the matching should be weakly better than staying unmatched for all agents. Second, no firm and group of workers can negotiate an agreement that is weakly better than the current matching for all members involved in the negotiations and strictly better for at least one member.

\begin{definition}\label{stability2}
\normalfont
A matching $(\mu,\mathbf{p})$ is \textbf{stable} if
\begin{description}
\item[(\romannumeral1)] (Individual Rationality) for each worker $w\in W$, $\mu(w)\in A_w\cup\{ø\}$ and $\mathrm{U}_w(\mu(w),\mathbf{p})\geq 0$; for each firm $f\in F$,  $\mu(f)\in \mathcal{A}_f\cup\{\emptyset\}$ and $\mathrm{U}_f(\mu(f),\mathbf{p})\geq 0$;

\item[(\romannumeral2)] (No Blocking Coalition) there are no firm-set of workers combination $(f,S)$ and price $\mathbf{p}'$ such that

    $\mathrm{U}_{f}(S,\mathbf{p}')\geq \mathrm{U}_{f}(\mu(f),\mathbf{p})$,\qquad\quad and

    $\mathrm{U}_w(f,\mathbf{p}')\geq \mathrm{U}_w(\mu(w),\mathbf{p})$\qquad for all $w\in S$

    with strict inequality holding for at least one member of $\{f\}\cup S$.
\end{description}
\end{definition}

Notice that since continuous transfers are allowed, and all agents have transferable utilities, the set of stable matchings coincides with the core. The inequalities in condition (ii) of Definition \ref{stability2} can be replaced with strict inequalities.

A hypergraph is a generalization of a graph in which an edge can join any number of vertices. A \textbf{firm-worker hypergraph} is a hypergraph $(N,\mathcal{E})$ in which each vertex $j\in N$ is either a firm or a worker, and each edge contains one firm and its possible employees:
$\mathcal{E}=\{\{f\}\cup S|f\in F \text{ and } S\in \mathcal{A}_f\}$.\footnote{Since we do not restrict wages to be nonnegative, an acceptable set for a firm will be chosen by the firm under proper wages. Hence, each acceptable set for a firm is the firm's possible set of employees.} A cycle is a cyclic alternating sequence of distinct vertices and edges: $(j^1, E^1, j^2, E^2, ..., j^k, E^k, j^1)$, where $k\geq 2$, $j^i,j^{i+1}\in E^i$ for each $i\in\{1,2,\ldots,k-1\}$, and $j^k,j^1\in E^k$.\footnote{The sequence $(j,E,j)(j\in E, E\in \mathcal{E})$ is not considered to be a cycle.} The number of edges contained in the cycle (i.e., the number $k$) is called the length of the cycle. We call $j^1,\ldots,j^k$ vertices of this cycle; other vertices contained in the edges of the cycle are not vertices of the cycle. A cycle is called a \textbf{nontrivial odd-length cycle} if it has an odd number of edges and every edge contains exactly two vertices of the cycle. 

\begin{theorem}\label{thm_exist1}
\normalfont
A stable matching exists in a matching market with transferable utilities if the firm-worker hypergraph has no nontrivial odd-length cycles.
\end{theorem}

\begin{example}\label{exam_A}
\normalfont
Consider a market with three firms, three workers, and the following profile of firms' valuations. Worker sets not listed are unacceptable for the firms.
\begin{equation}\label{exam_balancedA}
\begin{aligned}
&\mathrm{v}_{f_1}(\{w_1,w_2\})=3\quad &\mathrm{v}_{f_1}(\{w_1\})=1\qquad &\mathrm{v}_{f_1}(\emptyset)=0\\
&\mathrm{v}_{f_2}(\{w_1\})=2\quad &\mathrm{v}_{f_2}(\{w_3\})=1\qquad &\mathrm{v}_{f_2}(\emptyset)=0\\
&\mathrm{v}_{f_3}(\{w_2,w_3\})=2\quad &\mathrm{v}_{f_3}(\emptyset)=0\qquad
&\quad
\end{aligned}
\end{equation}
A hypergraph without nontrivial odd-length cycles is also called \textbf{balanced}. This market has the following balanced firm-worker hypergraph.
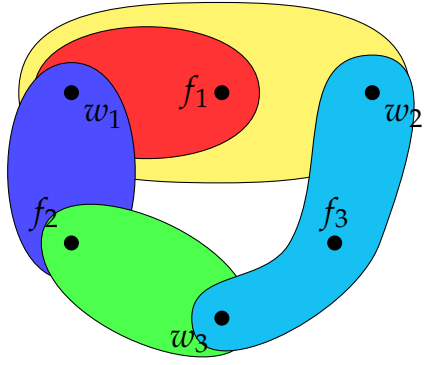
\begin{center}
\begin{tikzpicture}
    \node (v1) at (0,2) {};
    \node (v2) at (4,2) {};
    \node (v3) at (2,-1) {};
    \node (v4) at (2,2) {};
    \node (v5) at (0,0) {};
    \node (v6) at (3.5,0) {};
\begin{scope}[fill opacity=0.8]
    \filldraw[fill=yellow!70] ($(v1)+(-0.7,0)$)
        to[out=90,in=180] ($(v4)+(0,1.2)$)
        to[out=0,in=90] ($(v2)+(0.5,0)$)
        to[out=270,in=0] ($(v4)+(0,-1.2)$)
        to[out=180,in=270] ($(v1)+(-0.7,0)$);
    \filldraw[fill=red!80] ($(v1)+(-0.5,0)$)
        to[out=90,in=90] ($(v4)+(0.5,0)$)
        to[out=270,in=270] ($(v1)+(-0.5,0)$);
    \filldraw[fill=blue!70] ($(v1)+(0,0.4)$)
        to[out=0,in=0] ($(v5)+ (0,-0.5)$)
        to[out=180,in=180] ($(v1) + (0,0.4)$);
    \filldraw[fill=green!70] ($(v5)+(-0.4,0)$)
        to[out=90,in=90] ($(v3)+ (0.4,0)$)
        to[out=270,in=270] ($(v5) + (-0.4,0)$);
    \filldraw[fill=cyan!70] ($(v2)+(0,0.5)$)
        to[out=180,in=60] ($(v6)+ (-0.6,0)$)
        to[out=240,in=90] ($(v3) + (-0.4,0)$)
        to[out=270,in=250] ($(v6)+ (0.6,0)$)
        to[out=70,in=0] ($(v2) + (0,0.5)$);
\end{scope}
    \fill (v1) circle (0.1) node [below right] {$w_1$};
    \fill (v2) circle (0.1) node [below right] {$w_2$};
    \fill (v3) circle (0.1) node [below left] {$w_3$};
    \fill (v4) circle (0.1) node [left] {$f_1$};
    \fill (v5) circle (0.1) node [above left] {$f_2$};
    \fill (v6) circle (0.1) node [above] {$f_3$};
\end{tikzpicture}
\end{center}
\medskip
There is only one odd-length cycle:
\begin{equation}\label{odd_cycle} (w_1,\{w_1,f_1\},f_1,\{f_1,w_1,w_2\},w_2,\{w_2,f_3,w_3\},w_3,\{w_3,f_2\},f_2,\{f_2,w_1\},w_1)
\end{equation}
in which $\{f_1,w_1,w_2\}$ contains three vertices of the cycle. Therefore, a stable matching always exists in this market. For example, if we assume each worker cares only about her wage (i.e., $\mathrm{v}_w(f)=0$ for each $w$ and each $f$), we obtain a stable matching $(\mu,\mathbf{p})$ by specifying $\mu(f_1)=\emptyset$, $\mu(f_2)=\{w_1\}$, $\mu(f_2)=\{w_2,w_3\}$, and $\mathbf{p}=(2,1,1)$.
\end{example}

For any coalition $S\subseteq N$, let $\chi_S\in\{0,1\}^{N}$ denote the indicator vector of $S$.\footnote{The indicator vector $\chi_S$ of $S$ is a vector in $\{0,1\}^{N}$ such that $\chi_S(i)=1$ if $i\in S$, and $\chi_S(i)=0$ if $i\notin S$.} Let $\mathcal{E}$ denote the set of potential blocking coalitions: $\mathcal{E}\equiv\{\{f\}\cup S|f\in F, S\in\mathcal{A}_f, \text{ and } f\in A_w \text{ for each } w\in S\}$. Let $\mathcal{I}$ be the collection of all singleton coalitions: $\mathcal{I}\equiv\{\{i\}|i\in N\}$. For any $S\in\mathcal{I}\cup\mathcal{E}$, let $\mathrm{V}(S)$ be the \textbf{aggregate value obtained by coalition $S$}: $\mathrm{V}(S)=\mathrm{v}_f(S')+\sum_{i\in S'}\mathrm{v}_i(f)$ if $S=(\{f\}\cup S')\in \mathcal{E}$, and $\mathrm{V}(S)=0$ if $S\in \mathcal{I}$. 

\begin{proof}[Proof of Theorem \ref{thm_exist1}]
Let $\Psi$ be the collection of all $\mu$ functions in which no firm is matched with an unacceptable set, and no worker is matched with an unacceptable firm: $\Psi\equiv\{\mu|\mu(f)\in \mathcal{A}_f\cup\{\emptyset\}$ for each $f\in F$, and $\mu(w)\in A_w\cup\{ø\}$ for each $w\in W\}$. Let $\overline{V}\equiv\max_{\mu\in\Psi}\sum_{f\in F}\mathrm{V}(\{f\}\cup\mu(f))$ be the maximum aggregate value of all agents among all $\mu$ functions from $\Psi$.

Let $\mathbf{x}\in R^N$, and define $\mathrm{\widehat{x}}(S)\equiv\sum_{i\in S}\mathrm{x}(i)$ for each $S\in\mathcal{I}\cup\mathcal{E}$. Consider the following linear program where the decision variable is $\mathbf{x}\in \mathbb{R}^N$.
\begin{align}
&\text{minimize } & \sum_{i\in N}\mathrm{x}(i)\qquad\qquad\qquad\qquad\qquad\qquad\label{min}\\
&\text{subject to } & \mathrm{\widehat{x}}(S)\geq \mathrm{V}(S) \text{ for all }S\in\mathcal{I}\cup\mathcal{E}\quad\qquad\label{constraint1}
\end{align}

Let $\widetilde{X}$ be the value of (\ref{min}). Since each function $\mu$ from $\Psi$ partitions $N$ into coalitions from $\mathcal{I}\cup\mathcal{E}$, constraint (\ref{constraint1}) implies $\widetilde{X}\geq\overline{V}$. If $\widetilde{X}>\overline{V}$, there is no stable matching. If $\widetilde{X}=\overline{V}$, pick a vector $\widetilde{\mathbf{x}}\in \mathbb{R}^N$ that generates $\widetilde{X}$ and a function $\mu$ that generates $\overline{V}$. Suppose $\mu$ partitions $N$ into a set $\mathcal{S}\subset\mathcal{I}\cup\mathcal{E}$ of coalitions, then $\widetilde{X}=\overline{V}$ implies that constraint (\ref{constraint1}) for $\widetilde{\mathbf{x}}$ is binding for each $S\in\mathcal{S}$. We can thus construct a matching $(\mu,\mathbf{p})$ in which each agent $i\in N$ obtains utility $\widetilde{\mathrm{x}}(i)$. Notice that this matching is stable. Therefore, a stable matching exists if and only if the value of (\ref{min}) is $\overline{V}$. The following is the dual program where the decision variables are $\delta_S$ for $S\in \mathcal{I}\cup\mathcal{E}$.
\begin{align}
&\text{maximize } &\sum_{S\in \mathcal{I}\cup\mathcal{E}}\delta_S\mathrm{V}(S)\qquad\qquad\qquad\qquad\qquad\label{max}\\
&\text{subject to } & \sum_{S\in \mathcal{I}\cup\mathcal{E}}\delta_S\chi_S=\mathbf{1} \text{ and}\qquad\qquad\qquad\quad\label{polytope}\\
&\quad& \delta_S\geq0 \text{ for all } S\in \mathcal{I}\cup\mathcal{E}\qquad\qquad\quad\label{nonnegative}
\end{align}

According to the duality theorem (see, e.g., Theorem 4.9 of \citealp{V05}), the objective functions (\ref{min}) and (\ref{max}) have the same value. The above argument was used in proving the Bondareva-Shapley theorem. The difference is that we only concern coalitions that contain at most one firm. We know that a stable matching exists if and only if the value of (\ref{max}) does not exceed $\overline{V}$.

The incidence matrix of a hypergraph is a matrix whose rows and columns represent the vertices and edges of the hypergraph, respectively. An entry of the incidence matrix is 1 (or 0) if the vertex of the entry's row index is (resp. is not) in the edge of the entry's column index. A matrix is called balanced if it is the incidence matrix of a balanced hypergraph. If the firm-worker hypergraph is balanced, the coefficients from the left-hand side of (\ref{polytope}) form a balanced matrix. According to Lemma 2.1 of \cite{FHO74}, all vertices of the polytope defined by (\ref{polytope}) and (\ref{nonnegative}) are integral. Since the polytope is nonempty, the value of (\ref{max}) is attained by a group of integral $(\delta_S)_{S\in \mathcal{I}\cup\mathcal{E}}$ and therefore does not exceed $\overline{V}$.
\end{proof}

\section{Discrete matching\label{Sec_DM}}

This section studies basic discrete many-to-one matching. We mostly follow the notations of \cite{RS90}.

There is a finite set $F$ of firms and a finite set $W$ of workers. Let ${\o}$ be the null firm, representing not being matched with any firm. Each worker $w\in W$ has a strict preference ordering $\succ_w$ over $\widetilde{F}\equiv F\cup\{{\o}\}$. For any two firms $f, f'\in \widetilde{F}$, we write $f\succ_w f'$ when worker $w$ prefers firm $f$ to firm $f'$ according to $\succ_w$. We write $f\succeq_w f'$ if either $f\succ_w f'$ or $f=f'$. Let $\succ_W$ denote the preference profile of all workers. Each firm $f\in F$ has a strict preference ordering $\succ_f$ over $2^W$. For any two subsets of workers $S,S'\subseteq W$, we write $S\succ_f S'$ when firm $f$ prefers $S$ to $S'$ according to $\succ_f$. We write $S\succeq_f S'$ if either $S\succ_f S'$ or $S=S'$. Let $\succ_F$ be the preference profile of all firms. A matching market can be summarized as a tuple $\Gamma=(W,F,\succ_W,\succ_F)$.

Let $\mathrm{Ch}_f$ be the choice function of $f$ such that for any $S\subseteq W$, $\mathrm{Ch}_f(S)\subseteq S$ and $\mathrm{Ch}_f(S)\succeq_f S'$ for any $S'\subseteq S$. By convention, let $\mathrm{Ch}_{{\o}}(S)=S$ for all $S\subseteq W$. For any firm $f\in F$, any worker $w\in W$, and any nonempty set $S\subseteq W$, firm $f$ is called acceptable to worker $w$ if $f\succ_w {\o}$; worker set $S$ is called acceptable to firm $f$ if $S\succ_f\emptyset$; we say that $S$ is \textbf{satisfactory} to firm $f$ if $S=\mathrm{Ch}_f(S)$. For each firm $f\in F$, let $\mathcal{A}_f$ be the collection of firm $f$'s acceptable sets, and $\mathcal{Y}_f$ the collection of firm $f$'s satisfactory sets. We have $\mathcal{Y}_f\subseteq \mathcal{A}_f$: worker set $S$ is acceptable to firm $f$ if $S$ is satisfactory to $f$, but not vice versa. For example, if firm $f$ has the preference $\{w\}\succ\{w,w'\}\succ\emptyset$, then $\{w,w'\}$ is acceptable but not satisfactory to firm $f$. Firm $f$'s possible sets of employees are firm $f$'s satisfactory sets when we take $f$'s individual rationality into account.

\begin{definition}
\normalfont
A \textbf{matching} $\mu$ is a function from the set $\widetilde{F}\cup W$ into $\widetilde{F}\cup 2^W$ such that for all $f\in \widetilde{F}$ and $w\in W$,
\begin{description}
\item[(\romannumeral1)] $\mu(w)\in \widetilde{F}$;

\item[(\romannumeral2)] $\mu(f)\in 2^W$;

\item[(\romannumeral3)] $\mu(w)=f$ if and only if $w\in\mu(f)$.
\end{description}
\end{definition}

We say that a matching $\mu$ is \textbf{individually rational} if $\mu(w)\succeq_w {\o}$ for all $w\in W$ and $\mu(f)=\mathrm{Ch}_f(\mu(f))$ for all $f\in F$. Workers' individual rationalities mean that each matched worker prefers her current employer to being unmatched.  Firms' individual rationalities mean that no firm wish to unilaterally drop any of its employees. We say that a firm $f$ and a subset of workers $S\subseteq W$ form a \textbf{blocking coalition} that blocks $\mu$ if $f\succeq_w\mu(w)$ for all $w\in S$, and $S\succ_f\mu(f)$. A blocking coalition $(f,S)$ blocks matching $\mu$ if forming the coalition is weakly better for all members of the coalition and strictly better for at least one member. However, since preferences are strict and$\mu(f)\neq S$, firm $f$ must get strictly better.

\begin{definition}\label{stability}
\normalfont
A matching $\mu$ is \textbf{stable} if it is individually rational and there are no blocking coalitions that block $\mu$.\footnote{Since a firm unilaterally dropping some employees can be viewed as the firm and part of its employees forming a blocking coalition, the no-blocking-coalition condition implies firms' individual rationalities: $\mu(f)=\mathrm{Ch}_f(\mu(f))$ for all $f\in F$.}
\end{definition}

The notion of stability is stronger than the core in discrete matching. The set of stable matchings defined by Definition \ref{stability} is called the core defined by weak domination in \cite{RS90} and the set of stable$^*$ matchings in \cite{EO04}. Stable matchings are characterized as fixed points of certain operators; see \cite{A00}, \cite{F03}, \cite{EO04,EO06}, and \cite{HM05}, among others.

A firm-worker hypergraph is a hypergraph $(F\cup W,\mathcal{E})$ in which each vertex $j\in F\cup W$ is either a firm or a worker, and each edge contains one firm and its possible employees:
$\mathcal{E}=\{\{f\}\cup S|f\in F \text{ and } S\in \mathcal{Y}_f\}$. Notice that a set of a firm's possible employees is a satisfactory set for this firm in discrete matching.

\begin{theorem}\label{thm_exist2}
\normalfont
A stable matching exists in a discrete matching market if the firm-worker hypergraph has no nontrivial odd-length cycles.\footnote{See the definition for a nontrivial odd-length cycle in Section \ref{Sec_Mct}.}
\end{theorem}

\begin{example}\label{exam_B}
\normalfont
Consider the following firms' preference profile.
\begin{equation}\label{exam_balanced}
\begin{aligned}
&f_1: \{w_1,w_2\}\succ\{w_1\}\succ\emptyset\\
&f_2: \{w_1\}\succ\{w_3\}\succ\emptyset\\
&f_3: \{w_2,w_3\}\succ\emptyset
\end{aligned}
\end{equation}
This profile has the same firm-worker hypergraph as market (\ref{exam_balancedA}). Hence, Theorem \ref{thm_exist2} indicates that a stable matching exists for all possible preferences of workers. The odd-length cycle (\ref{odd_cycle}) in this market seems to be an ``instability cycle'' as the cycle of market (\ref{exam_in1}) if we assume the following workers' preferences.
\begin{equation*}
w_1: f_1\succ f_2 \qquad\qquad w_2: f_3\succ f_1 \qquad\qquad w_3: f_2\succ f_3
\end{equation*}
We can then construct a similar cyclic blocking process as illustrated in market (\ref{exam_in1}). However, this cycle breaks down if we match firm $f_1$ with $\{w_1,w_2\}$ and match firm $f_2$ with $\{w_3\}$ (or match $f_2$ with $\{w_1$\} and match $f_3$ with $\{w_2,w_3\}$).
\end{example}

We have an immediate corollary from Theorem \ref{thm_exist2}: A marriage market always has a stable matching. As illustrated in Section \ref{Sec_cyclic}, a woman-man graph of any marriage market has no odd-length cycles since each cycle of a woman-man graph consists of several man-woman-man paths. Theorem \ref{thm_exist2} follows from a more general statement below.

\begin{proposition}\label{propexist}
\normalfont
A stable matching exists in a discrete matching market if the firm-worker hypergraph has no nontrivial odd-length cycles satisfying the condition: for each $f\in F$ and each pair $S,S'$ such that both $\{f\}\cup S$ and $\{f\}\cup S'$ are edges of the cycle, $S=\mathrm{Ch}_f(S\cup S')$ or $S'=\mathrm{Ch}_f(S\cup S')$.
\end{proposition}

This statement follows from H23a's integer feasibility program and the existence of a stable schedule matching showed by \cite{CKK19}. In a schedule matching, each worker schedules her time among different firms, and each firm schedules its time among different groups of workers. H23a showed that a stable schedule matching can always be transformed into a stable (full-time) matching if there is an integral vertex on a nonempty polytope $\{\mathbf{z}\mid B\mathbf{z}=\mathbf{1}, \mathbf{z}\geq0\}$ in which $B$ is a 0-1 matrix, and each column of $B$ with at least two ones is an edge of the firm-worker hypergraph. See Section 3 of H23a for the construction of the polytope. According to Theorem 2 of \cite{CKK19}, a stable schedule matching always exists, then Lemma 2.1 of \cite{FHO74} implies that a stable (full-time) matching exists if matrix $B$ is balanced. Moreover, for each firm $f\in F$ and each pair $S,S'$ such that both $\{f\}\cup S$ and $\{f\}\cup S'$ are columns of $B$, $S=\mathrm{Ch}_f(S\cup S')$ or $S'=\mathrm{Ch}_f(S\cup S')$ holds. This is because $S$ and $S'$ are the sets consumed in a consumption process of H23a. We therefore obtain Proposition \ref{propexist}.

\begin{example}\label{exam_C}
\normalfont
Consider the firms' preference profile
\begin{equation}\label{exam_prop1}
f_1: \{w_1,w_2\}\succ\emptyset \qquad\qquad\qquad f_2: \{w_1,w_2\}\succ\{w_1\}\succ\{w_2\}\succ\emptyset
\end{equation}
Although the firm-worker hypergraph of (\ref{exam_prop1}) has the nontrivial odd-length cycle depicted in Section \ref{Sec_intro}, this preference profile satisfies the condition of Proposition \ref{propexist}: The odd-length cycle contains $\{f_2,w_1\}$ and $\{f_2,w_2\}$, however, neither $\{w_1\}=\mathrm{Ch}_{f_2}(\{w_1,w_2\})$ nor $\{w_2\}=\mathrm{Ch}_{f_2}(\{w_1,w_2\})$ holds.
\end{example}

Proposition \ref{propexist} is clearly a generalization of Theorem \ref{thm_exist2}. It is also a generalization of H23a's total unimodularity condition. We prove the last statement in the Appendix (see Section \ref{Sec_prop1}).

We can also prove Theorem \ref{thm_exist1} and Theorem \ref{thm_exist2} using Scarf's lemma since the problems of matching with transferable utilities and discrete matching are special NTU-games.\footnote{NTU-games refer to coalitional games with nontransferable utilities.} The two problems correspond to similar polytopes in Scarf's lemma,\footnote{See the polytope in Section 6 of \cite{H23b} for discrete matching, which is essentially the same as the polytope defined by (\ref{polytope}) and (\ref{nonnegative}) in our Section \ref{Sec_Mct} for matching with transferable utilities.} and therefore, it is not surprising to see similar existence results in the two problems. However, it is more direct to prove Theorem \ref{thm_exist1} using the LP duality for the Bondareva-Shapley theorem, and we cannot derive Proposition \ref{propexist} merely from the properties of the polytope in Scarf's lemma. In a subsequent work (\citealp{H23b}), based on a reduction of \cite{K10} from the set of stable matchings to the core, we proposed a concavity condition on the whole preference profile of firms and workers. This result is an analogue to the balancedness condition of \cite{S67} and subsumes Theorem \ref{thm_exist2}. In contrast, the conditions in Theorem \ref{thm_exist2} and Proposition \ref{propexist} are on firms' preferences and are easier to verify.

\section{Application\label{Sec_app}}

This section presents an application of our condition based on a structure called technology roadmap, which generalizes the structure of H23a's technology tree. A \textbf{technology roadmap} $TR=(T,W)$ is a directed tree\footnote{A tree is a connected acyclic graph. A directed tree is a directed acyclic graph whose underlying undirected graph is a tree.} $T=(V,E)$ defined on a set of workers $W$. Each vertex from $V$ represents a \textbf{technology}. Each technology $v\in V$ demands a subset of workers $W^v\subseteq W$ to implement. We say that worker $w$ engages in technology $v$ if $w\in W^v$. Each edge is an upgrade from one technology to another. The technology roadmap generalizes H23a's technology tree in two aspects: (i) A technology tree is a rooted tree, whereas a technology roadmap does not necessarily have a root. (ii) An upgrade in a technology tree requires more workers such that the worker set expands along with the edge direction, whereas a technology roadmap does not impose restrictions on worker sets demanded by the technologies.\footnote{We will restrict worker sets demanded by the technologies by requiring workers to be specialists. However, the definition of technology roadmap itself does not impose restrictions on worker sets demanded by the technologies.} In a technology tree of H23a, there is only one evolutionary path toward a technology. However, the technology roadmap allows multiple evolutionary paths toward a technology. We illustrate this with the following example.

\begin{example}\label{exam_app}
\normalfont
Consider a market with three firms and five workers. A technology roadmap is depicted as follows.
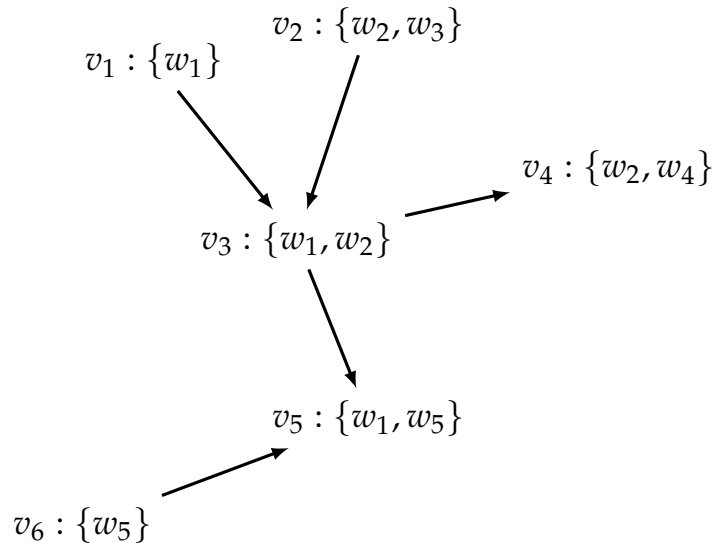
\begin{center}
\begin{tikzpicture}[scale=0.95]

  \node (n1) at (1,6.5) {$v_1:\{w_1\}$};
  \node (n2) at (4,7) {$v_2:\{w_2,w_3\}$};
  \node (n3) at (7.5,5)  {$v_4:\{w_2,w_4\}$};
  \node (n4) at (3,4) {$v_3:\{w_1,w_2\}$};
  \node (n5) at (4,1.5)  {$v_5:\{w_1,w_5\}$};
  \node (n6) at (0,0){$v_6:\{w_5\}$};

  \foreach \from/\to in {n1/n4,n2/n4,n4/n3,n4/n5,n6/n5}
    \draw [-latex](\from) -- (\to) [very thick];

\end{tikzpicture}
\end{center}

Each vertex from $\{v_1,\ldots,v_6\}$ represents a technology. A technology can have multiple evolutionary paths: Technology $v_3$ can be upgraded from technology $v_1$ or $v_2$; technology $v_5$ can be upgraded from technology $v_3$ or $v_6$. Each technology requires the set of workers on the right to implement.
\end{example}

H23a assumed each worker to be a specialist over a technology tree. We instead assume that both firms and workers are specialized in a technology roadmap. Specializations of workers and firms are captured by the scopes of technologies they engage in and possess, respectively. In a technology roadmap $TR=(T,W)$ with $T=(V,E)$, a subgraph $P=(V^P,E^P)$ of $T$ is called a \textbf{technology path} if $(V^P,E^P)$ is either a directed path\footnote{A directed path in a directed graph is a path whose edges are all oriented in the same direction.} or a vertex (i.e., $|V^P|=1$ and $E^P=\emptyset$). We say that worker $w$ engages in the subgraph $G^w=(V^w,E^w)$ of $T$ in which $V^w=\{v|w\in W^v\}$ and $E^w=\{vv'|v,v'\in V^w\text{ and }vv'\in E\}$.

\begin{definition}\label{def_specialist}
\normalfont
Worker $w$ is called a \textbf{specialist} in technology roadmap $TR$ if $G^w$ is a technology path.
\end{definition}

We assume each worker is a specialist who engages in a subgraph that is a technology path. For instance, in the above technology roadmap, worker $w_1$ engages in the subgraph $v_1\rightarrow v_3\rightarrow v_5$. Worker $w_2$ engages in the subgraph $v_2\rightarrow v_3\rightarrow v_4$. Other workers are also specialists.

\begin{definition}\label{def_specialized}
\normalfont
In a market $\Phi$ with transferable utilities or a discrete matching market $\Gamma$, firms are \textbf{specialized} over a technology roadmap $
TR=(T,W)$ if there is a collection of technology paths $\{P^f=(V^f,E^f)\}_{f\in F}$ of $T$ such that
\begin{description}
\item[(\romannumeral1)]for each $f\in F$, $S\in \mathcal{A}_f$ implies $S=W^v$ for some $v\in V^f$;

\item[(\romannumeral2)]$V^f\cap V^{f'}=\emptyset$ for any $f,f'\in F$ with $f\neq f'$.
\end{description}
\end{definition}

Condition (i) requires each firm's acceptable sets to be from a technology path of a technology roadmap. Condition (ii) requires different firms' acceptable sets to be from disjoint technology paths. For instance, consider the following discrete matching market.
\begin{equation}\label{pre_app}
\begin{aligned}
&f_1: \{w_2,w_4\}\succ\{w_1\}\succ  \emptyset\\
&f_2: \{w_2,w_3\}\succ \emptyset\\
&f_3: \{w_1,w_5\}\succ \{w_5\}\succ \emptyset
\end{aligned}
\end{equation}
Firm $f_1$'s acceptable sets are from the technology path $v_1\rightarrow v_3\rightarrow v_4$ of the above technology roadmap; firm $f_2$'s acceptable sets are from the technology path $v_2$; and firm $f_3$'s acceptable sets are from the technology path $v_5\rightarrow v_6$. We can regard this assumption as that firm $f$ possesses the technologies of technology path $P^f$ for each $f\in F$. Each firm has an arbitrary preference over the technologies it possesses, which induces the firm's preference over the worker sets demanded by the technologies.\footnote{We allow a technology possessed by a firm to be unacceptable for this firm. For example, if firm $f_1$ possesses technologies $v_1$, $v_3$, and $v_4$ from the technology roadmap in Example \ref{exam_app}, the firm can have the preference showed in (\ref{pre_app}).}

Although condition (ii) implies that acceptable sets of different firms are induced from different technologies, firms can still have common acceptable sets since the same set of workers can implement multiple technologies. The collection of technology paths in Definition \ref{def_specialized} may not be unique. For example, suppose there is only one firm with only one acceptable set $\{w\}$. The firm is specialized over the technology roadmap $v_1:\{w\}\rightarrow v_2:\{w\}$ in which the firm may possess $v_1$ or $v_2$.

\begin{theorem}\label{thm_special}
\normalfont
In a market $\Phi$ with transferable utilities or a discrete matching market $\Gamma$, the firm-worker hypergraph has no nontrivial odd-length cycles if there exists a technology roadmap in which (i) each worker is a specialist, and (ii) the firms are specialized over this technology roadmap.
\end{theorem}

The proof of this theorem is relegated into the appendix (see Section \ref{Append_thm}). To see the roles of the specializations of workers and firms, consider the following two roadmaps.
\bigskip
\begin{center}
\begin{tikzpicture}

  \node (n1) at (0,1.8) {$v_1:\{w_1\}$};
  \node (n2) at (3.5,1.8) {$v_2:\{w_2\}$};
  \node (n3) at (7,1.8)  {$v_3:\{w_1,w_2\}$};
  \node (n4) at (0,0) {$v_1:\{w_1\}$};
  \node (n5) at (3.5,0)  {$v_2:\{w_1,w_2\}$};
  \node (n6) at (7,0){$v_3:\{w_2\}$};

  \foreach \from/\to in {n1/n2,n2/n3,n4/n5,n5/n6}
    \draw [-latex] (\from) -- (\to) [very thick];

\end{tikzpicture}
\end{center}

In the first roadmap, worker $w_1$ engages in the subgraph that consists of two isolated vertices $v_1$ and $v_3$. This roadmap can induce the firm-worker hypergraph in Section \ref{Sec_intro} when firm $f_1$ has a single acceptable set from the technology path $v_3$ and firm $f_2$'s acceptable sets are from the technology path $v_1\rightarrow v_2$. The workers are specialists in the second roadmap, but firms are not specialized in the roadmap if firm $f_1$ has a single acceptable set $\{w_1,w_2\}$, and firm $f_2$ has two acceptable sets $\{w_1\}$ and $\{w_2\}$.\footnote{This case may be induced from the technology roadmap if (i) firm $f_1$ possesses the technology $v_2$, and firm $f_2$ possesses the technology $v_1$ and $v_3$, or (ii) firm $f_1$ possesses the technology $v_2$, and firm $f_2$ possesses all technologies from the technology path $v_1\rightarrow v_2\rightarrow v_3$. The firms are not specialized over this roadmap in either (i) or (ii).} This case also gives rise to the cycle of the firm-worker hypergraph in Section \ref{Sec_intro}.

\section{Appendix}

\subsection{Proof of Theorem \ref{thm_special}}\label{Append_thm}
Suppose there exists a technology roadmap $TR=(T,W)$ in which each worker is a specialist and the firms are specialized over this technology roadmap, and there exists the following cycle of length $k\geq3$ in the firm-worker hypergraph:
\begin{equation}\label{app_cycle}
(j(1), E^1, j(2), E^2, ..., j(k), E^k, j(1))
\end{equation}
where $k$ is odd, and no edge of the cycle contains more than two vertices of the cycle.

For each $i\in\{1,2,\ldots,k\}$, let $P^{j(i)}$ be the directed path of $T$ for $j(i)$ defined in Definition \ref{def_specialist} if $j(i)$ is a worker; let all $P^{j(i)}$ with $j(i)\in F$ be from one collection of technology paths defined in Definition \ref{def_specialized}. Based on this collection of technology paths, we can use different technologies of $T$ to represent different edges of the firm-worker hypergraph: Since no two firms possess a common technology, edge $E^i$ that involves firm $f$ corresponds to a vertex on firm $f$'s technology path.\footnote{Recall that the firm-worker hypergraph in discrete matching is defined on firms' satisfactory sets. Since each firm's satisfactory sets are also acceptable to the firm, each edge of the firm-worker hypergraph corresponds to a vertex of the roadmap.} Thus, we use the term vertex $E^i$ to mean the vertex corresponding to $E^i$ on $T$ for each $i\in\{1,\ldots,k\}$. Since $j(1)$ belongs to both $E^k$ and $E^1$, vertex $E^k$ and vertex $E^1$ both lie on $P^{j(1)}$. Without loss of generality, we assume that the directed path between vertex $E^k$ and vertex $E^1$ points from vertex $E^k$ to vertex $E^1$.

Since $T$ is a tree in which any two vertice are connected by exactly one path, the other half of cycle (\ref{app_cycle}), $(E^1, j(2), E^2, ..., j(k), E^k)$, must go through vertex $E^1$ to reach vertex $E^k$. Thus, there must be a segment of this half of cycle going through vertex $E^1$ in $P^{j(1)}$ and in the opposite direction of $P^{j(1)}$. Let $(E^t,j(t+1),E^{t+1})$ with $t\in\{1,\ldots,k-1\}$ be such a segment. Since vertex $E^1$ is on this segment, and this segment is in $P^{j(t+1)}$, we know that $E^1$ contains $j(t+1)$. If $t\neq 1$, then $E^1$ contains $j(1)$, $j(2)$, and $j(t+1)$. This contradicts that (\ref{app_cycle}) is a nontrivial odd-length cycle. Hence, $t=1$, and we have $t+1=2\neq k$ since $k\geq3$. If a directed path starts from one vertex and ends at another vertex, we call the former vertex the initial vertex of the directed path, and the latter vertex the terminal vertex of the directed path. We know that vertex $E^1$ is the terminal vertex of the $E^1E^2$ segment of $P^{j(2)}$.\footnote{This is because the $E^1E^2$ segment of $P^{j(2)}$ goes through vertex $E^1$ in the opposite direction of $P^{j(1)}$.}

If vertex $E^2$ lies on $P^{j(1)}$, then $E^2$ contains $j(1)$, $j(2)$, and $j(3)$. A contradiction. If vertex $E^k$ lies on $P^{j(2)}$, then $E^k$ contains $j(k)$, $j(1)$, and $j(2)$. A contradiction. Hence, the $E^1E^2$ segment of $P^{j(2)}$ is a directed path from vertex $E^2$ to vertex $E^1$ and joining the $E^kE^1$ segment of $P^{j(1)}$ at some vertex other than vertex $E^k$ or $E^1$. We thus know that vertex $E^2$ and vertex $E^k$ are connected by a path that is not a directed path. Since vertex $E^{k-1}$ and vertex $E^k$ both lie on $P^{j(k)}$ and are thus connected by a directed path, we know $k-1\neq2$, and thus $k\geq5$. The structure of the technology roadmap is illustrated below.\footnote{This figure shows one possible structure for the roadmap under the assumption that there is cycle (\ref{app_cycle}) in the firm-worker hypergraph. There are also other possible structures for the roadmap.}
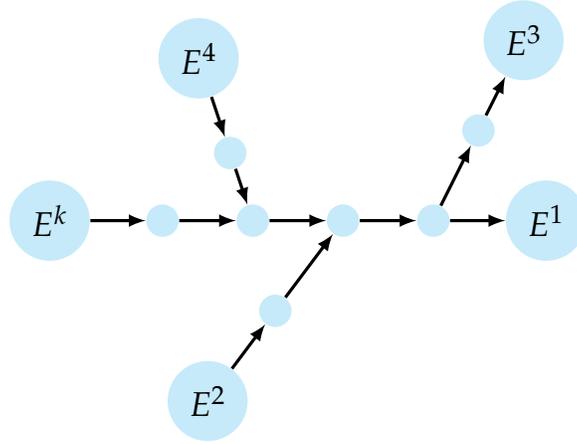
\begin{center}
\tikzstyle{c}=[circle,fill=cyan!20]
\begin{tikzpicture}[scale=0.6]

  \node (Ek) at (-4.5,0) [c] {$E^k$};
  \node (E1) at (6.5,0) [c] {$E^1$};
  \node (E2) at (-1,-4) [c] {$E^2$};
  \node (E3) at (6,4) [c] {$E^3$};
  \node (E4) at (-1.2,3.6) [c] {$E^4$};
  \node (n1) at (-2,0) [c] {};
  \node (n2) at (0,0) [c] {};
  \node (n3) at (2,0) [c] {};
  \node (n4) at (4,0) [c] {};
  \node (n5) at (-0.5,1.5) [c] {};
  \node (n6) at (0.5,-2) [c] {};
  \node (n7) at (5,2) [c] {};

  \foreach \from/\to in {Ek/n1,n1/n2,n2/n3,n3/n4,n4/E1}
    \draw [-latex] (\from) -- (\to) [very thick];

  \foreach \from/\to in {E4/n5,n5/n2,n4/n7,n7/E3,E2/n6,n6/n3}
    \draw [-latex] (\from) -- (\to) [very thick];

\end{tikzpicture}
\end{center}
Let $T^t$ be the union of the segment $E^kE^1$ of $P^{j(1)}$ and all the segments $E^{x-1}E^x$ of $P^{j(x)}$ for all $x\in \{2,\ldots,t\}$. A leaf of a tree is a vertex incident to only one edge. We assume inductively that, for some $s\in\{2,\ldots,k-3\}$ and $s$ is even, $T^s$ is a directed tree that has $E^k,E^1,E^2,\ldots,E^s$ as leaves where $E^x$ with odd $x\in\{1,\ldots,s-1\}$ are the terminal vertices of their incident edges, and $E^k$ and $E^x$ with even $x\in\{2,\ldots,s\}$ are the initial vertices of their incident edges. Notice that the inductive assumption holds for $s=2$.

Since vertex $E^k$ and vertex $E^s$ are connected by a path in $T^s$ where $E^s$ is a leaf and the initial vertex of its incident edge, the other half of the cycle (\ref{app_cycle}), $(E^s, j(s+1), E^{s+1}, ..., j(k), E^k)$, must go through vertex $E^s$ in $P^{j(s)}$ and in the same direction of $P^{j(s)}$. Thus, we know that there is $(E^x,j(x+1),E^{x+1})$ with $x\in\{s,\ldots,k-1\}$ such that vertex $E^s$ lies on $P^{j(x+1)}$. If $x\neq s$, then $E^s$ contains $j(s)$, $j(s+1)$, and $j(x+1)$. A contradiction. Hence, $x=s$ and we know that vertex $E^s$ is the initial vertex of the $E^sE^{s+1}$ segment of $P^{j(s+1)}$.

If vertex $E^{s+1}$ lies on $P^{j(x)}$ with $x\in\{1,\ldots,s\}$, then $E^{s+1}$ contains $j(s+1)$, $j(s+2)$, and $j(x)$. A contradiction. If vertex $E^x$ with $x\in\{1,\ldots,s-1\}\cup\{k\}$ lies on $P^{j(s+1)}$, then $E^{x}$ contains $j(x)$, $j(x+1)$ (or $j(1)$ if $x=k$), and $j(s+1)$. A contradiction. Hence, the $E^sE^{s+1}$ segment of $P^{j(s+1)}$ leaves $T^s$ at some vertex other than $E^1,E^2,\ldots,E^s$ or $E^k$. Thus, $T^{s+1}$ is a directed tree that has $E^k,E^1,E^2,\ldots,E^{s+1}$ as leaves where $E^x$ with odd $x\in\{1,\ldots,s+1\}$ are the terminal vertices their incident edges, and $E^k$ and $E^x$ with even $x\in\{2,\ldots,s\}$ are the initial vertices of their incident edges.

Since vertex $E^k$ and vertex $E^{s+1}$ are connected by a path in $T^{s+1}$ where $E^{s+1}$ is a leaf and the terminal vertex of its incident edge, the other half of the cycle (\ref{app_cycle}), $(E^{s+1}, j(s+2), E^{s+2}, ..., j(k), E^k)$, must go through vertex $E^{s+1}$ in $P^{j(s+1)}$ and in the opposite direction of $P^{j(s+1)}$. Thus, we know that there is $(E^y,j(y+1),E^{y+1})$ with $y\in\{s+1,\ldots,k-1\}$ such that vertex $E^{s+1}$ lies on $P^{j(y+1)}$. If $y\neq s+1$, then $E^{s+1}$ contains $j(s+1)$, $j(s+2)$, and $j(y+1)$. A contradiction. Hence, $y=s+1$ and we know that vertex $E^{s+1}$ is the terminal vertex of the $E^{s+1}E^{s+2}$ segment of $P^{j(s+1)}$.

If vertex $E^{s+2}$ lies on $P^{j(x)}$ with $x\in\{1,\ldots,s+1\}$, then $E^{s+2}$ contains $j(s+2)$, $j(s+3)$, and $j(x)$. A contradiction. If vertex $E^x$ with $x\in\{k\}\cup\{1,\ldots,s\}$ lies on $P^{j(s+2)}$, then $E^{x}$ contains $j(x)$, $j(x+1)$ (or $j(1)$ if $x=k$), and $j(s+2)$. A contradiction. Hence, the $E^{s+1}E^{s+2}$ segment of $P^{j(s+2)}$ joins $T^{s+1}$ at some vertex other than $E^1,E^2,\ldots,E^{s+1}$ or $E^k$. Thus, $T^{s+2}$ is a directed tree that has $E^k,E^1,E^2,\ldots,E^{s+2}$ as leaves where $E^x$ with odd $x\in\{1,\ldots,s+1\}$ are the terminal vertices of their incident edges, and $E^k$ and $E^x$ with even $x\in\{2,\ldots,s+2\}$ are the initial vertices of their incident edges.

Recall that the inductive assumption holds for $s=2$. According to the above inductive argument, for any even $t\leq k-1$, $T^t$ is a directed tree that has $E^k,E^1,E^2,\ldots,E^t$ as leaves where $E^x$ with odd $x\in\{1,\ldots,t-1\}$ are the terminal vertices of their incident edges, and $E^k$ and $E^x$ with even $x\in\{2,\ldots,t\}$ are the initial vertices of their incident edges. Then, $E^t$ is a leaf of $T^t$ and the initial vertex of its incident edge, and $E^k$ is also a leaf of $T^t$ and the initial vertex of its incident edge; thus, $E^t$ and $E^k$ are connected by a path that is not a directed path. Since vertex $E^{k-1}$ and vertex $E^k$ are connected by a directed path in $P^{j(k)}$, we know that $k-1$ is not even, and $k$ is not odd. A contradiction. Therefore, the firm-worker hypergraph has no nontrivial odd-length cycles.

\subsection{On Proposition \ref{propexist} and totally unimodular demand types\label{Sec_prop1}}

A matrix is totally unimodular if every square submatrix has determinant 0 or $\pm1$. In particular, each entry in a totally unimodular matrix is 0 or $\pm1$. A set of vectors is called totally unimodular if the matrix that has these vectors as columns is totally unimodular. H23a adopted the notion of demand type from \cite{BK19} to represent how a firm's demand changes as its available set expands.

\begin{definition}\label{df_dt}
\normalfont
For each firm $f\in F$, let $\mathscr{D}_f=\{\mathbf{d}\in \{0,1\}^W\mid  \mathbf{d}\neq \mathbf{0}$ and $\mathbf{d}=\chi_{\mathrm{Ch}_f(S)}-\chi_{\mathrm{Ch}_f(S')}$ for some $S,S'$ with $S'\subset S\subseteq W\}$ be firm $f$'s demand type. The demand type for the firms' preference profile is $\mathscr{D}=\cup_{f\in F}\mathscr{D}_f$.
\end{definition}

H23a showed that a stable matching exists when the firms' demand type $\mathscr{D}$ is totally unimodular. For example, consider the firms' preference profile in Example \ref{exam_C}. Firm $f_1$ has a demand type of $\{(1,1)\}$ since $f_1$ will hire both workers when its available set expands from $\emptyset$ to $\{w_1,w_2\}$. Firm $f_2$ has a demand type of $\{(1,1),(1,0),(0,1)\}$. Hence, the firms' demand type $\mathscr{D}$ is $\{(1,1),(1,0),(0,1)\}$, which is totally unimodular (because any matrix formed by two of these vectors has determinant 0 or $\pm1$).

\begin{proposition}\label{TU}
\normalfont
The firm-worker hypergraph satisfies the condition of Proposition \ref{propexist} if the firms' demand type $\mathscr{D}$ is totally unimodular.
\end{proposition}

\begin{proof}
We provide the proof first, and then an example to illustrate the proof. Suppose the firm-worker hypergraph contains a nontrivial odd-length cycle satisfying the condition described in Proposition \ref{propexist}. Let $M$ be the incidence matrix of the cycle. Since the cycle is a nontrivial odd-length cycle, we know that $M$ is a matrix of odd order with exactly two 1s in each row and column. Thus, we have $|M|=2$ or $-2$. For each firm $f\in F$, let $M_f$ be the submatrix of $M$ that consists of columns from $M$ where each column of $M_f$ corresponds to an edge of the firm-worker hypergraph that contains firm $f$. Let $m^1_f,m^2_f,\cdots$ denote the columns of $M_f$, and let $S^1_f, S^2_f,\cdots$ denote $f$'s satisfactory sets corresponding to $m^1_f,m^2_f,\cdots$, respectively. According to the property of the cycle described in Proposition \ref{propexist}, we assume $S^j_f=\mathrm{Ch}_f(S_f^j\cup S_f^k)$ if $j>k$.
For each firm $f$ involved in the cycle, let $M'_f$ be the matrix that consists of columns $m^1_f$, $m^2_f-m^1_f$, $m^3_f-m^1_f$,$\cdots$, and $M'$ the matrix that consists of all columns of all $M'_f$ for $f\in F$.

We have $|M'|=|M|$ since subtracting one column from another column leaves the determinant unchanged. For each  firm $f$ involved in the cycle, if column $m^1_f$ contains a component $1$ for firm $f$, we remove both the column and the row containing this component from $M'$. Let $M''$ denote the remaining matrix $M''$. Then, each column of $M''$ is from the firms' demand type $\mathscr{D}$. We have  $|M''|=2$ or $-2$,\footnote{We have $|M''|=|M'|$ or $|M''|=-|M'|$ since each row we remove contains exactly one $1$ and these $1$s lie on a diagonal of $M'$ after permutating the rows and columns of $M'$ properly.} and thus, the firms' demand type $\mathscr{D}$ is not totally unimodular.
\end{proof}

For example, let $M$ be the incidence matrix for the cycle in market (\ref{exam_in1}).
\begin{center}
\begin{tabular}
[c]{c|cccc}
& $\{f_1,w_1\}$ & $\{f_1,w_2\}$ & $\{w_1,w_2\}$\\\hline
$f_1$ & $1$ & $1$ & $0$\\
$w_1$ & $1$ & $0$ & $1$\\
$w_2$ & $0$ & $1$ & $1$%
\end{tabular}
\end{center}
We then have
\begin{equation*}
M_{f_1}=
\begin{pmatrix}
1      & 1     \\
0      & 1     \\
1      & 0
\end{pmatrix}\text{ in which }
m^1_{f_1}=
\begin{pmatrix}
1      \\
0      \\
1
\end{pmatrix}\text{ and }
m^2_{f_1}=
\begin{pmatrix}
1      \\
1      \\
0
\end{pmatrix},\quad
M_{f_2}=
\begin{pmatrix}
0      \\
1      \\
1
\end{pmatrix},
\end{equation*}
and thus,
\begin{equation*}
M'_{f_1}=
\begin{pmatrix}
1      & 0     \\
0      & 1     \\
1      & -1
\end{pmatrix},\quad
M'_{f_2}=
\begin{pmatrix}
0      \\
1      \\
1
\end{pmatrix},\quad
\begin{tabular}
[c]{c|cccc}
& $\{f_1,w_1\}$ & $\{f_1,w_2\}$ & $\{w_1,w_2\}$\\\hline
$f_1$ & $1$ & $0$ & $0$\\
$w_1$ & $0$ & $1$ & $1$\\
$w_2$ & $1$ & $-1$ & $1$%
\end{tabular}=M'.
\end{equation*}
We remove the first row and the first column from $M'$, both containing the component for $f_1$. We obtain the remaining matrix
\begin{equation*}
M''=
\begin{pmatrix}
1     & 1     \\
-1      & 1
\end{pmatrix}.
\end{equation*}
The firms' demand type $\mathscr{D}$ is $\{(1,1),(1,0),(0,1),(1,-1)\}$ in market (\ref{exam_in1}). The columns of $M''$ are from $\mathscr{D}$.

\subsection{On the reverse implication of the existence theorems}\label{necessity}

In this section, we show that our condition is not necessary in the sense stated in Section \ref{Sec_cyclic}. We provide below a firm-worker hypergraph with a nontrivial odd-length cycle and show that, in both matching with transferable utilities and discrete matching, stable matchings exist in all markets with this firm-worker hypergraph.

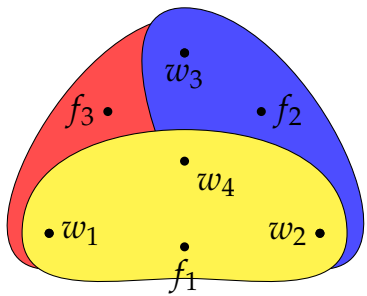
\normalfont
\begin{center}
\begin{tikzpicture}[scale=0.6]
    \node (v1) at (-1,0) {};
    \node (v2) at (2,4) {};
    \node (v3) at (5,0) {};
    \node (v4) at (2,0.8) {};
    \node (v5) at (2,-1) {};
    \node (v6) at (2,-0.3) {};
    \node (v7) at (2,1.6) {};
    \node (v8) at (0.3,2.7) {};
    \node (v9) at (3.7,2.7) {};

    \begin{scope}[fill opacity=0.7]
    \filldraw[fill=red!70] ($(v2)+(0,0.8)$)
        to[out=0,in=60] ($(v4)$)
        to[out=240,in=0] ($(v1) + (0,-0.8)$)
        to[out=180,in=180] ($(v2) + (0,0.8)$);
    \filldraw[fill=blue!70] ($(v2)+(0,1)$)
        to[out=180,in=120] ($(v4)$)
        to[out=300,in=180] ($(v3) + (0,-0.8)$)
        to[out=0,in=0] ($(v2) + (0,1)$);
    \filldraw[fill=yellow!80] ($(v1)+(-0.6,0)$)
        to[out=90,in=180] ($(v4)+ (0,1.5)$)
        to[out=0,in=90] ($(v3) + (0.6,0)$)
        to[out=270,in=0] ($(v5)$)
        to[out=180,in=270] ($(v1) + (-0.6,0)$);
    \end{scope}

    \foreach \v in {1,2,3} {
        \fill (v\v) circle (0.1);
    }

    \fill (v1) circle (0.1) node [right] {$w_1$};
    \fill (v2) circle (0.1) node [below] {$w_3$};
    \fill (v3) circle (0.1) node [left] {$w_2$};
    \fill (v7) circle (0.1) node [below right] {$w_4$};
    \fill (v6) circle (0.1) node [below] {$f_1$};
    \fill (v8) circle (0.1) node [left] {$f_3$};
    \fill (v9) circle (0.1) node [right] {$f_2$};
\end{tikzpicture}
\end{center}
\medskip
This hypergraph has one nontrivial odd-length cycle:
\begin{equation*} (w_1,\{w_1,f_1,w_4,w_2\},w_2,\{w_2,f_2,w_4,w_3\},w_3,\{w_3,f_3,w_4,w_1\},w_1)
\end{equation*}
Consider a market with this firm-worker hypergraph in the framework of matching with transferable utilities. Suppose there is at least one edge gaining a positive aggregate value, otherwise all agents staying unmatched is a stable matching. Then, we can always obtain a stable matching by matching workers to the firm within the edge that has the largest aggregate value and letting all agents obtain zero utility except for worker $w_4$.

A discrete matching market with this firm-worker hypergraph always has a stable matching: If all firms are acceptable to all workers, matching $w_4$ to her favorite firm (together with two other workers) is a stable matching. It is also straightforward to see that a stable matching exists if some firms are not acceptable to some workers.

\end{document}